\documentclass[11pt]{article}

\newcommand{\cmp}{Comm. Math. Phys.~}

\newcommand{\jpa}{J. Phys. A~}

\newcommand{\prl}{Phys. Rev. Lett.~}
\newcommand{\pra}{Phys. Rev. A~}

\usepackage{young}
\usepackage[latin1]{inputenc}
\usepackage[linesnumbered,boxed,ruled,commentsnumbered]{algorithm2e}
\usepackage{appendix}
\usepackage{epstopdf}
\usepackage{xcolor}
\usepackage{subfigure}
\usepackage[T1]{fontenc}
\usepackage[sc]{mathpazo}
\usepackage{amsmath}
\usepackage{amssymb}
\usepackage{graphicx,color}
\usepackage{framed}
\usepackage{multirow}
\usepackage{enumerate}
\usepackage{amsthm}% theorems and proofs
\usepackage{amsfonts,mathrsfs}
\usepackage{geometry} % allows easy specifying of the page layout
\usepackage{eepic}
\usepackage{ifthen}
\usepackage[vcentermath]{youngtab}
\usepackage[unicode=true,pdfusetitle, bookmarks=true,bookmarksnumbered=false,bookmarksopen=false, breaklinks=false,pdfborder={0 0 0},backref=false,colorlinks=false] {hyperref}
\hypersetup{
colorlinks,linkcolor=myurlcolor,citecolor=myurlcolor,urlcolor=myurlcolor}
\definecolor{myurlcolor}{rgb}{0,0,0.7}% Color highlighted for numbers in citations and theorems

\usepackage{color}

\newcommand{\blue}{\textcolor{blue}}

\newcommand{\proj}[1]{| #1\rangle\!\langle #1 |}

\usepackage{hyperref}
\hypersetup{pdfpagemode=UseNone}

\newcommand{\tinyspace}{\mspace{1mu}}

\newcommand{\op}[1]{\operatorname{#1}}

\newcommand{\abs}[1]{\left\lvert\tinyspace #1 \tinyspace\right\rvert}

\newcommand{\norm}[1]{\left\lVert\tinyspace #1 \tinyspace\right\rVert}

\renewcommand{\det}{\operatorname{det}}

\newcommand{\setft}[1]{\mathrm{#1}}
\newcommand{\lin}[1]{\setft{L}\left(#1\right)}

\def \dif {\mathrm{d}}
\def \diag {\mathrm{diag}}

\def\I{\mathbb{1}}

\newenvironment{mylist}[1]{\begin{list}{}{
    \setlength{\leftmargin}{#1}
    \setlength{\rightmargin}{0mm}
    \setlength{\labelsep}{2mm}
    \setlength{\labelwidth}{8mm}
    \setlength{\itemsep}{0mm}}}
    {\end{list}}

\def\ot{\otimes}

\newcommand{\out}[2]{| #1\rangle\langle #2 |}
\newcommand{\Inner}[2]{\left\langle #1 , #2\right\rangle}

\newcommand{\Herm}{\mathrm{Herm}}

\newcommand{\Pa}[1]{\left(#1\right)}

\newcommand{\set}[1]{\{#1\}}
\newcommand{\Set}[1]{\left\{#1\right\}}

\newcommand{\ket}[1]{|#1\rangle}

\DeclareMathOperator{\trace}{Tr}

\newcommand{\Ptr}[2]{\trace_{#1}\Pa{#2}}

\newcommand{\Tr}[1]{\Ptr{}{#1}}

\newcommand{\Abs}[1]{\left|\tinyspace#1\tinyspace\right|}

\def\cD{\mathcal{D}}

\def\cK{\mathcal{K}}\def\cL{\mathcal{L}}\def\cN{\mathcal{N}}\def\cO{\mathcal{O}}
\def\cR{\mathcal{R}}\def\cT{\mathcal{T}}
\def\cX{\mathcal{X}}

\def\bsA{\boldsymbol{A}}\def\bsB{\boldsymbol{B}}\def\bsD{\boldsymbol{D}}
\def\bsF{\boldsymbol{F}}\def\bsH{\boldsymbol{H}}
\def\bsK{\boldsymbol{K}}\def\bsO{\boldsymbol{O}}

\def\bsU{\boldsymbol{U}}\def\bsV{\boldsymbol{V}}\def\bsW{\boldsymbol{W}}\def\bsX{\boldsymbol{X}}\def\bsY{\boldsymbol{Y}}
\def\bsZ{\boldsymbol{Z}}
\def\bsa{\boldsymbol{a}}

\def\bsr{\boldsymbol{r}}
\def\bsx{\boldsymbol{x}}\def\bsy{\boldsymbol{y}}

\def\rD{\mathrm{D}}
\def\rH{\mathrm{H}}
\def\rO{\mathrm{O}}

\def\rU{\mathrm{U}}

\def\sO{\mathscr{O}}
\def\sR{\mathscr{R}}\def\sT{\mathscr{T}}

\def\bbC{\mathbb{C}}

\def\bbR{\mathbb{R}}

\newtheorem{thrm}{Theorem}[section]
\newtheorem{lem}[thrm]{Lemma}
\newtheorem{prop}[thrm]{Proposition}
\newtheorem{cor}[thrm]{Corollary}
\theoremstyle{definition}

\numberwithin{equation}{section}

\newcounter{questionnumber}

\begin{document}

%========================================================================================%
\title{Spectral moments and entropy rigidity of quantum channels: a three-mode photonic witness}
%========================================================================================%

\author{\blue{Laizhen Luo}$^1$\footnote{E-mail: luolaizhen@126.com}\quad and\quad \blue{Lin Zhang}$^2$\footnote{E-mail: godyalin@163.com}\\
  {\it\small $^1$Foundation Science Education Center, Hangzhou City University, Hangzhou 310015, China}\\
{\it \small $^2$School of Mathematical Sciences, Hangzhou Dianzi University, Hangzhou 310018, China}}

\date{}
\maketitle
\begin{abstract}

Majorization and von Neumann entropy provide related but
inequivalent descriptions of spectral disorder in quantum optical
states. We investigate the channels preserving these descriptions
and propose a three-mode photonic test of their difference. For
channels with equal input and output dimensions, we show that
constant output purity and third spectral moment on the unitary
orbit of one simple-spectrum state determine the
majorization-preserving channel forms, without assuming unitality.
For qubits, output purity alone is sufficient. Preservation of
entropy equality is substantially more restrictive: in dimension
$d\geqslant3$, even a local condition near the maximally mixed state
permits only replacement and unitary channels. Qubits admit an
additional depolarizing family, but this exception disappears under
extension by an idle auxiliary qubit. We construct exactly
entropy-matched qutrit inputs whose output entropies split at cubic
order under nontrivial depolarizing noise. A proposed single-photon
implementation uses three optical modes, randomized Weyl
transformations and population or spectral-moment measurements.
Explicit expressions are also obtained for preparation mismatch,
readout uncertainty and input-independent loss. The results connect
quadratic and cubic spectral information with experimentally
interpretable tests of quantum-channel structure.
\\~\\
\noindent\textbf{Keywords:} quantum optical channels, majorization,
von Neumann entropy, spectral moments, photonic qutrits
\end{abstract}

\newpage

%=============================================================================%
\section{Introduction}%1
%=============================================================================%

Quantum optical devices modify not only amplitudes and coherences
but also the eigenvalue distributions of density operators. Spectral
information is therefore relevant to mixed-state preparation, noise
characterization and entropy-based comparisons of optical outputs.
Majorization and von Neumann entropy are two widely used
descriptions of spectral disorder, but they impose different
requirements on a physical transformation.

Majorization compares complete eigenvalue distributions
\cite{Marshall1979,Nielsen2001}. In quantum information theory, it
characterizes deterministic transformations between bipartite pure
states under local operations and classical communication
\cite{Nielsen1999}. It also provides a natural framework for
studying mixing and spectral order. The associated channel questions
are closely connected to matrix-preserver problems
\cite{Ando1989,Hiai1987}.

Entropy compresses spectral information into one scalar. A channel
can consequently preserve majorization while mapping two
equal-entropy inputs to outputs with different entropies. This
distinction is absent at the spectral level for qubits, whose
entropy determines their unordered pair of eigenvalues. It appears
already for qutrits, where different spectra can have the same
entropy.

Three-mode photonic systems provide a minimal optical setting for
investigating this difference. Controlled mode populations, phase
shifts and multiport transformations allow the preparation and
processing of single-photon qutrits \cite{Reck1994,Clements2016}. In addition, direct
estimation of nonlinear state functionals offers alternatives to
full state tomography when only spectral information is required
\cite{Ekert2002}.

In this work, we address two structural questions. First, how much
spectral information is needed to determine whether a channel
preserves majorization? Second, which channels preserve equality of
entropy, without necessarily preserving its numerical value?

For the first question, we show that constant output purity and
third moment on the complete unitary orbit of one simple-spectrum
input determine the channel form. Unitality is not an independent
assumption: unless the channel is replacement, it follows from the
quadratic condition. For qubits, purity alone is sufficient.

For the second question, we establish a local rigidity result. In
equal input and output dimensions $d\geqslant3$, a channel that
preserves entropy equality near the maximally mixed state must be
replacement or unitary. Qubits admit an additional depolarizing
family because the cubic trace vanishes on their traceless Hermitian
space. Requiring preservation after extension by an idle auxiliary
qubit removes this exception.

We then construct an explicit pair of entropy-matched qutrits whose
entropies separate at cubic order under depolarizing noise. The
proposed optical implementation uses diagonal single-photon states
and randomized Weyl operations. Exact finite-parameter expressions,
a model-based error budget and an input-independent loss model
clarify how the effect can be tested.

The optical proposal illustrates a channel property rather than
nonclassicality: the selected probes commute. It is also not a
complete certification protocol for an arbitrary unknown channel.
These distinctions are important when interpreting the theoretical
results experimentally.

%====================================================%
\section{Spectral order and entropy levels}%2
%====================================================%

\subsection{States and channels}

Let $\lin{\bbC^d}$ denote the complex $d\times d$ matrices. The
state space is
\begin{eqnarray*}
\rD(\bbC^d) = \set{\rho\in \lin{\bbC^d}: \rho\geqslant0,
\Tr{\rho}=1}
\end{eqnarray*}
We write
\begin{eqnarray*}
\tau_d=\frac{\I_d}{d}, \qquad \Herm_0(\bbC^d)
=\Set{\bsX=\bsX^\dagger:\Tr{\bsX}=0}.
\end{eqnarray*}
The real vector space $\Herm_0(\bbC^d)$ carries the Hilbert-Schmidt
inner product $\Inner{\bsX}{\bsY}=\Tr{\bsX\bsY}$. Unless stated
otherwise, a quantum channel is a completely positive
trace-preserving linear map $\Phi:\lin{\bbC^d} \to
\lin{\bbC^d}(d\geqslant2)$. Thus the input and output dimensions are
equal. Standard channel terminology and the Choi representation are
used throughout \cite{Watrous2017}.

All logarithms are natural, so entropy is measured in nats:
\begin{eqnarray}\label{eq:entropy}
S(\rho) = -\Tr{\rho\ln \rho}.
\end{eqnarray}

%2.1
\subsection{Majorization and spectral equivalence}

For real vectors $\bsx,\bsy\in\bbR^d$, the relation $\bsx\prec \bsy$
means
\begin{eqnarray*}
\sum_{j=1}^{k}x_j^\downarrow \leqslant \sum_{j=1}^{k}y_j^\downarrow,
\qquad 1\leqslant k<d,
\end{eqnarray*}
together with equality of the total sums. For Hermitian matrices,
majorization is defined through their eigenvalues:
\begin{eqnarray*}
\bsA\prec \bsB \quad\Longleftrightarrow\quad
\lambda(\bsA)\prec\lambda(\bsB).
\end{eqnarray*}
Equivalently \cite{Marshall1979,Ando1989,Watrous2017}, $\bsA\prec
\bsB$ if and only if
\begin{eqnarray}\label{eq:2}
\bsA=\sum_j p_j\bsV_j\bsB\bsV_j^\dagger
\end{eqnarray}
for some probabilities $p_j$ and unitaries $\bsV_j$. We write
$\rho\sim\sigma$ when $\rho$ and $\sigma$ have the same eigenvalue
multiset. Equivalently, $\rho=\bsV\sigma \bsV^\dagger$ for some
unitary $\bsV$, or both $\rho\prec\sigma$ and $\sigma\prec\rho$. A
channel is \emph{majorization preserving} if
\begin{eqnarray}\label{eq:3}
\rho\prec\sigma \quad\Longrightarrow\quad
\Phi(\rho)\prec\Phi(\sigma).
\end{eqnarray}
It preserves entropy levels if
\begin{eqnarray}\label{eq:equalentropy}
S(\rho)=S(\sigma) \quad\Longrightarrow\quad
S(\Phi(\rho))=S(\Phi(\sigma)).
\end{eqnarray}
Condition \eqref{eq:equalentropy} does not require
$S(\Phi(\rho))=S(\rho)$. It requires only that an input entropy
level be mapped into a single output entropy level.

For $d\geqslant 3$, equal-entropy states can have different spectra.
Such states are incomparable by majorization because entropy is
strictly Schur-concave. Consequently, preservation of majorization
does not immediately imply preservation of entropy levels.

\subsection{Relevant channel families}

Three channel families will appear:
\begin{eqnarray}\label{eq:4}
\begin{aligned}
\cR_\omega(\bsX)
&=\Tr{\bsX}\omega,\\
\cD_{a,\bsU}(\bsX) &=a\bsU\bsX\bsU^\dagger
 +(1-a)\Tr{\bsX}\tau_d,\\
\cT_{a,\bsU}(\bsX) &=a\bsU\bsX^\mathsf{T}\bsU^\dagger
 +(1-a)\Tr{\bsX}\tau_d,
\end{aligned}
\end{eqnarray}
where $\omega\in\rD(\bbC^d)$, $\bsU$ is unitary, and transposition
is taken in a fixed orthonormal basis. The complete-positivity
intervals are
\begin{eqnarray}\label{eq:5}
-\frac{1}{d^2-1}\leqslant a\leqslant1 \quad\text{for }\cD_{a,\bsU},
\end{eqnarray}
and
\begin{eqnarray}\label{eq:6}
-\frac{1}{d-1}\leqslant a\leqslant\frac{1}{d+1} \quad\text{for
}\mathcal \cT_{a,\bsU}.
\end{eqnarray}
The replacement family contains arbitrary fixed output states and
need not be unital. The other two families are unital.

%====================================================%
\section{Majorization preservation from a single input orbit}
%====================================================%

\subsection{A characterization using two spectral moments}

For a state $\rho_0$, define its unitary orbit by $\cO_{\rho_0} =
\Set{\bsV\rho_0\bsV^\dagger:\bsV\in \rU(d)}$. The global normal
forms in the next theorem produce the majorization-preserver
classification. The additional statements strengthen the
single-orbit result in that work by removing unitality and reducing
the spectral information required.

\begin{thrm}[Single-orbit moment characterization]\label{th:1}
Let $d\geqslant2$, let $\Phi$ be a quantum channel, and fix a state
$\rho_0\in\rD(\bbC^d)$ with distinct eigenvalues. The following
statements are equivalent:
\begin{enumerate}[(1)]
\item $\Phi$ preserves majorization on $\rD(\bbC^d)$.
\item $\Phi$ preserves spectral equivalence:
$\rho\sim\sigma\Longrightarrow\Phi(\rho)\sim\Phi(\sigma)$.
\item The output spectrum is constant on $\cO_{\rho_0}$.
\item For $d\geqslant3$, both
\begin{eqnarray}\label{eq:7}
\Tr{\Phi(\bsV\rho_0\bsV^\dagger)^2} \quad\text{and}\quad
\Tr{\Phi(\bsV\rho_0\bsV^\dagger)^3}
\end{eqnarray}
are independent of $\bsV$. For $d=2$, constancy of the second moment
alone is sufficient.
\item $\Phi$ belongs to one of the families in Eq.~\eqref{eq:4}, with parameters
satisfying Eq.~\eqref{eq:5} or \eqref{eq:6}.
\end{enumerate}
\end{thrm}
The input state need not have full rank. Only simplicity of its
spectrum is required.

The second and third moments play different roles. The second moment
forces a scaled orthogonal action on traceless perturbations and,
except for replacement channels, unitality. The third moment
identifies the matrix-algebra structure of that action.

\subsection{Low-degree identities on a unitary orbit}

The proof uses the following lemma.
\begin{lem}[Quadratic and cubic rigidity on a simple-spectrum
orbit]\label{lem:1} Let $\bsA\in\Herm_0(\bbC^d)$ have distinct
eigenvalues, and define
$\cO_{\bsA}=\set{\bsV\bsA\bsV^\dagger:\bsV\in \rU(d)}$.
\begin{enumerate}[(1)]
\item Suppose $g$ is a homogeneous quadratic polynomial and
$\ell$ is a linear functional on $\Herm_0(\bbC^d)$. If $g+\ell$ is
constant on $\cO_{\bsA}$, then
\begin{eqnarray}\label{eq:8}
\ell=0,   \qquad   g(X)=\gamma\Tr{\bsX^2}
\end{eqnarray}
for some real $\gamma$.
\item If $d\geqslant3$ and a homogeneous cubic polynomial $f$ is
constant on $\cO_{\bsA}$, then
\begin{eqnarray}\label{eq:9}
f(X)=\beta\Tr{\bsX^3}
\end{eqnarray}
for some real $\beta$.
\end{enumerate}
\end{lem}
A proof is given in Appendix~\ref{app:A}. The first statement is the
affine-quadratic extension needed to remove the unitality
assumption.

\subsection{Proof of Theorem~\ref{th:1}}

The implications $1\Rightarrow2\Rightarrow3\Rightarrow4$ are
immediate. In particular, if $\rho\sim\sigma$, majorization holds in
both directions, and the same must then hold for their outputs. We
prove $4\Rightarrow5$.
\begin{itemize}
\item \textbf{Step 1: Constant purity forces replacement or
unitality.} Set
\begin{eqnarray*}
\bsA=\rho_0-\tau_d, \qquad \bsB=\Phi(\tau_d), \qquad
\sT=\Phi|_{\Herm_0(\bbC^d)}.
\end{eqnarray*}
Trace preservation gives $\sT(\Herm_0(\bbC^d))\subseteq
\Herm_0(\bbC^d)$, and $\Phi(\tau_d+\bsX)=\bsB+\sT(\bsX)$. The output
purity is
\begin{eqnarray*}
\Tr{\Pa{\bsB+\sT(\bsX)}^2} = \Tr{\bsB^2} + 2\Tr{\bsB\sT(\bsX)} +
\Tr{\sT(\bsX)^2}.
\end{eqnarray*}
By assumption, this expression is constant for $\bsX\in\cO_{\bsA}$.
Lemma~\ref{lem:1} therefore gives
\begin{eqnarray}\label{eq:10}
\Tr{\bsB\sT(\bsX)}=0
\end{eqnarray}
and
\begin{eqnarray}\label{eq:11}
\Tr{\sT(\bsX)^2} = \gamma\Tr{\bsX^2}
\end{eqnarray}
for all $\bsX\in\Herm_0(\bbC^d)$. Since both sides of
Eq.~\eqref{eq:11} are nonnegative, $\gamma\geqslant0$. If
$\gamma=0$, then $\sT=0$, so $\Phi(\bsX)=\Tr{\bsX}\bsB$. Thus $\Phi$
is a replacement channel.

Suppose $\gamma>0$. Polarization of Eq.~\eqref{eq:11} shows that
\begin{eqnarray}\label{eq:12}
\sT=c\sO, \qquad c=\sqrt{\gamma}>0,
\end{eqnarray}
where $\sO$ is an orthogonal transformation of $\Herm_0(\bbC^d)$. In
particular, $\sT$ is surjective. Eq.~\eqref{eq:10} then implies that
$\bsB$ is orthogonal to every traceless Hermitian matrix. Since
$\Tr{\bsB}=1$,
\begin{eqnarray}\label{eq:13}
\bsB=\tau_d.
\end{eqnarray}
Hence every nonreplacement channel satisfying the single-orbit
purity condition is automatically unital.
\item \textbf{Step 2: The third moment determines the orthogonal
part.} Assume $d\geqslant3$. Using Eq.~\eqref{eq:13},
\begin{eqnarray*}
\Tr{\Pa{\tau_d+\sT(\bsX)}^3} = \frac{1}{d^2} +
\frac{3}{d}\Tr{\sT(\bsX)^2} + \Tr{\sT(\bsX)^3}.
\end{eqnarray*}
The second moment is already constant on $\cO_{\bsA}$. Therefore the
third-moment hypothesis implies that $\bsX\mapsto\Tr{\sT(\bsX)^3}$
is constant on that orbit. By Lemma~\ref{lem:1} and
Eq.~\eqref{eq:12}, there is a real number \(\eta\) such that
\begin{eqnarray}\label{eq:14}
\Tr{\sO(\bsX)^3} = \eta\Tr{\bsX^3} \qquad \text{for all }\bsX\in
\Herm_0(\bbC^d).
\end{eqnarray}
The cubic polynomial $\Tr{\bsX^3}$ is nonzero on $\Herm_0(\bbC^d)$
for $d\geqslant3$. Define
$$
M_d= \max_{\Tr{\bsX^2}=1} \Abs{\Tr{\bsX^3}}>0.
$$
Because $\sO$ maps the Hilbert-Schmidt unit sphere onto itself,
equation Eq.~\eqref{eq:14} gives $M_d=\abs{\eta}M_d$. Thus
$\eta=\pm1$. Set $\sR=\eta \sO$. Then $\sR$ is orthogonal and
$\Tr{\sR(\bsX)^3}=\Tr{\bsX^3}$.
\item \textbf{Step 3: Recovery of a Jordan automorphism.} Extend $\sR$ to all Hermitian matrices by
\begin{eqnarray}\label{eq:15}
\widehat\sR(s\I_d+\bsX)=s\I_d+\sR(\bsX), \qquad \bsX\in
\Herm_0(\bbC^d).
\end{eqnarray}
This map is unital, trace preserving, and Hilbert-Schmidt
orthogonal. It also preserves the cubic trace:
$\Tr{\widehat\sR(\bsH)^3} = \Tr{\bsH^3}$. Differentiating this
identity at $\bsH$ in a Hermitian direction $\bsK$ yields
$\Tr{\widehat\sR(\bsH)^2\widehat\sR(\bsK)} = \Tr{\bsH^2\bsK}$.
Orthogonality gives $\Tr{\widehat\sR(\bsH^2)\widehat\sR(\bsK)} =
\Tr{\bsH^2\bsK}$. Since $\widehat\sR$ is surjective,
$\widehat\sR(\bsH^2)=\widehat\sR(\bsH)^2$. Thus $\widehat\sR$ is a
Jordan automorphism of the Hermitian matrix algebra. Its standard
classification gives $\widehat\sR(\bsH)=\bsU\bsH\bsU^\dagger$ or
$\widehat\sR(\bsH)=\bsU\bsH^\mathsf{T}\bsU^\dagger$ for a fixed
unitary $\bsU$. Since $\sT=c\sO=\eta c\sR$, setting $a=\eta c$ gives
$\Phi(\bsX) = a\widehat\sR(\bsX) + (1-a)\Tr{\bsX}\tau_d$. Complex
linearity extends the formula from Hermitian matrices to all
matrices. For $d=2$, every orthogonal transformation of the Bloch
vector is induced either by unitary conjugation or by unitary
conjugation after transposition. Consequently, the quadratic
condition already gives the same channel forms.

\item \textbf{Step 4: Complete positivity.} Let
$\ket{\Omega}=\sum_{j=1}^d\ket{jj}$ and
$\bsF=\sum_{i,j=1}^d\out{ij}{ji}$ be the unnormalized maximally
entangled vector and the swap operator. We use the Choi convention
$C_\Phi=(\Phi\otimes\operatorname{id}) (\proj{\Omega})$. An output
unitary does not affect complete positivity. For the depolarizing
branch,
\begin{eqnarray}\label{eq:16}
C_{\cD_{a,\I}} = a\proj{\Omega} + \frac{1-a}{d}\I_{d^2}.
\end{eqnarray}
Its eigenvalues are $\frac{1+(d^2-1)a}{d}$ and $\frac{1-a}{d}$, with
multiplicities $1$ and $d^2-1$, respectively. Positivity is
therefore equivalent to Eq.~\eqref{eq:5}.

For the transpose-depolarizing branch,
\begin{eqnarray}\label{eq:17}
C_{\cT_{a,\I}} = a\bsF+\frac{1-a}{d}\I_{d^2}.
\end{eqnarray}
The eigenvalues on the symmetric and antisymmetric subspaces are
$\frac{1+(d-1)a}{d}$ and $\frac{1-(d+1)a}{d}$, which give
Eq.~\eqref{eq:6} \cite{Choi1975}.

Finally, $5\Rightarrow1$ follows from Eq.~\eqref{eq:2}. Replacement
channels produce identical outputs. For the other two families,
unitary conjugation of the input induces unitary conjugation of the
output:
\begin{eqnarray*}
\cD_{a,\bsU}(\bsV\rho \bsV^\dagger) &=& (\bsU\bsV\bsU^\dagger)\cD_{a,\bsU}(\rho)(\bsU\bsV\bsU^\dagger)^\dagger,\\
\cT_{a,\bsU}(\bsV\rho \bsV^\dagger) &=&
(\bsU\overline{\bsV}\bsU^\dagger)\cT_{a,\bsU}(\rho)(\bsU\overline
{\bsV}\bsU^\dagger)^\dagger.
\end{eqnarray*}
Applying these identities to a mixed-unitary representation of
$\rho\prec\sigma$ proves output majorization. This argument also
covers negative admissible values of $a$. \qed
\end{itemize}

\subsection{Why purity alone is insufficient $d\geqslant3$}

The third moment in Theorem~\ref{th:1} cannot generally be omitted
when $d\geqslant3$. Choose an orthogonal map $\sO$ on
$\Herm_0(\bbC^d)$ that is not, even up to an overall sign, the
restriction of a Jordan automorphism. Define
\begin{eqnarray}\label{eq:18}
\Phi_c(\bsX) = \Tr{\bsX}\tau_d + c\sO\Pa{\bsX-\Tr{\bsX}\tau_d},
\end{eqnarray}
where $\sO$ is extended complex linearly. At $c=0$, the Choi matrix
is strictly positive. Therefore $\Phi_c$ is completely positive for
all sufficiently small $\abs{c}$. It is also trace preserving and
unital.

For every state,
\begin{eqnarray}\label{eq:19}
\Tr{\Phi_c(\rho)^2} = \frac1d + c^2\Pa{\Tr{\rho^2}-\frac1d}.
\end{eqnarray}
Thus the output purity is constant on every input unitary orbit. For sufficiently small nonzero $c$, however, the channel is not one of the forms in Theorem~\ref{th:1}.

Purity detects isotropic contraction of the traceless space,
but the cubic moment is needed to identify the matrix-algebra
structure.

%====================================================%
\section{Local rigidity of entropy-level preservation}
%====================================================%

We now turn from spectral order to entropy levels.  A channel is locally \emph{entropy-level preserving} at $\tau_d$ if there exists a neighborhood $\cN$ of $\tau_d$ such that
\begin{eqnarray}\label{eq:20}
S(\rho)=S(\sigma) \quad\Longrightarrow\quad
S(\Phi(\rho))=S(\Phi(\sigma)) \qquad (\rho,\sigma\in\cN).
\end{eqnarray}
Although this condition is imposed only near one state, it determines the form of the channel on its entire input space.

\begin{thrm}[Entropy-level classification]\label{th:2}
Let $d\geqslant2$, and let $\Phi$ be a quantum channel with equal
input and output dimensions. The following are equivalent:
\begin{enumerate}[(1)]
\item $\Phi$ is preserves entropy levels on the full state space.
\item $\Phi$ is locally entropy-level preserving at $\tau_d$.
\item $\Phi$ belongs to the following classes: (i) For $d\geqslant3$:
\begin{eqnarray}\label{eq:21}
\Phi(\bsX)=\Tr{\bsX}\omega
   \quad\text{or}\quad
   \Phi(\bsX)=\bsU\bsX\bsU^\dagger.
\end{eqnarray}
(ii) For $d=2$: $\Phi$ is a replacement channel $\Phi(\bsX)=\Tr{\bsX}\omega$, or
\begin{eqnarray}\label{eq:22}
\Phi(\bsX)= a\bsU\bsX\bsU^\dagger
   +(1-a)\Tr{\bsX}\frac{\I_2}{2},
\qquad
   -\frac13\leqslant a\leqslant1.
\end{eqnarray}
\end{enumerate}
Here $\omega$ is a fixed state and $\bsU$ a fixed unitary.
\end{thrm}

\subsection{Reduction to an analytic scalar relation}

Assume the local condition \eqref{eq:20}, set
$$
\bsB=\Phi(\tau_d), \quad \cK=\operatorname{supp}(\bsB), \quad
r=\dim\cK,\quad \sT=\Phi|_{\Herm_0(\bbC^d)}.
$$
For every input state $\rho$,
\begin{eqnarray}\label{eq:23}
0\leqslant\rho\leqslant d\tau_d
\implies
0\leqslant\Phi(\rho)\leqslant d\bsB.
\end{eqnarray}
Consequently, every output is supported on $\cK$, as is $\sT(\bsX)$ for every $\bsX\in
\Herm_0(\bbC^d)$. Although $\bsB$ may be singular in the original
output space, it is strictly positive on $\cK$. Therefore
$\bsX\mapsto S(\bsB+\sT(\bsX))$ is real analytic for sufficiently
small $\bsX$. We claim that there exists a real-analytic function
$f$, defined near $\ln d$, such that
\begin{eqnarray}\label{eq:24}
S(\bsB+\sT(\bsX)) = f(S(\tau_d+\bsX))
\end{eqnarray}
for all sufficiently small $\bsX\in\Herm_0(\bbC^d)$.

To establish this, choose
$\bsD=\frac1{\sqrt2}\diag(1,-1,0,\ldots,0)$. The states
$\tau_d+t\bsD$ and $\tau_d-t\bsD$ have the same spectrum. Hence
$S(\tau_d+t\bsD)=E(t^2)$ and $S(\bsB+t\sT(\bsD))=G(t^2)$, where $E$
and $G$ are analytic. The second equality follows from the local
entropy-level condition, which makes the output entropy an even
function of $t$.

Moreover, $E(0)=\ln d$ and $E'(0)=-\frac{d}2\neq0$. The analytic
inverse function theorem gives $f=G\circ E^{-1}$.

Every sufficiently small entropy deficit from $\ln d$ is attained on
this reference curve. Matching an arbitrary nearby state to a state
on the curve and applying the local condition proves
Eq.~\eqref{eq:24}. No regularity assumption on $f$ was required in
advance.

\subsection{The entropy Hessian forces replacement or unitality}

For Hermitian operators on $\cK$, define
\begin{eqnarray}\label{eq:25}
\cL_{\bsB}(\bsY) = \int_0^\infty (\bsB+s\I_{\cK})^{-1} \bsY
(\bsB+s\I_{\cK})^{-1}\dif s.
\end{eqnarray}
The quadratic form $\Tr{\bsY\cL_{\bsB}(\bsY)}$ is strictly positive
for nonzero Hermitian $\bsY$.

The first derivative of entropy vanishes at $\tau_d$, and $D^2
S_{\tau_d}[\bsX,\bsY] = -d\Tr{\bsX\bsY}$. Writing $\alpha=f'(\ln d)$, differentiating
Eq.~\eqref{eq:24} twice gives
\begin{eqnarray}\label{eq:26}
\Tr{\sT(\bsX)\cL_{\bsB}(\sT(\bsY))} = \alpha d\Tr{\bsX\bsY}.
\end{eqnarray}
Thus $\alpha\geqslant0$. If $\alpha=0$, positivity of the quadratic
form implies $\sT=0$. The channel is then a replacement channel.

Suppose $\alpha>0$. Equation Eq.~\eqref{eq:26} makes $\sT$
injective. Since its image consists of traceless Hermitian operators
supported on $\cK$, $d^2-1\leqslant r^2-1$. As $r\leqslant d$, we
obtain $r=d$, and $\sT$ is surjective on $\Herm_0(\bbC^d)$. The
first derivative of Eq.~\eqref{eq:24} gives
\begin{eqnarray}\label{eq:27}
\Tr{\sT(\bsX)\ln \bsB}=0 \qquad \text{for all }\bsX\in
\Herm_0(\bbC^d).
\end{eqnarray}
Surjectivity implies that $\ln \bsB$ is proportional to the
identity. Hence $\bsB=\tau_d$.

At the maximally mixed state, equation Eq.~\eqref{eq:26} becomes
$\Tr{\sT(\bsX)\sT(\bsY)} = \alpha\Tr{\bsX\bsY}$. Consequently,
\begin{eqnarray}\label{eq:28}
\sT=c\sO, \qquad c=\sqrt{\alpha}>0,
\end{eqnarray}
with $\sO$ orthogonal on $\Herm_0(\bbC^d)$. Thus local entropy-level
preservation has already forced a nonreplacement channel to be
unital and isotropic on traceless perturbations.

\subsection{The cubic term fixes the contraction for $d\geqslant3$}

The entropy expansion around $\tau_d$ is
\begin{eqnarray}\label{eq:29}
S(\tau_d+\bsX) = \ln d -\frac{d}2\Tr{\bsX^2}
+\frac{d^2}{6}\Tr{\bsX^3} +O(\norm{\bsX}_2^4).
\end{eqnarray}
The imput entropy deficit begins at second order. Therefore the nonlinear
terms in $f$ in Eq.~\eqref{eq:24} contribute only at fourth and
higher orders. Comparing the cubic terms yields $\Tr{\sT(\bsX)^3} =
c^2\Tr{\bsX^3}$. Using Eq.~\eqref{eq:28},
\begin{eqnarray}\label{eq:30}
\Tr{\sO(\bsX)^3} = \frac1c\Tr{\bsX^3}.
\end{eqnarray}

For $d\geqslant3$, the quantity $M_d$ introduced in the proof of
Theorem~\ref{th:1} is positive. Taking maxima over the
Hilbert-Schmidt unit sphere in Eq.~\eqref{eq:30} gives
$M_d=\frac{M_d}{c}$, and hence
\begin{eqnarray}\label{eq:31}
c=1.
\end{eqnarray}
The channel is therefore unital and Hilbert-Schmidt isometric. This
isometry extends from Hermitian matrices to all matrices by complex
linearity.

A unital completely positive map satisfies the Schwarz inequality
\begin{eqnarray*}
\Phi(\bsZ^\dagger \bsZ) \geqslant \Phi(\bsZ)^\dagger\Phi(\bsZ).
\end{eqnarray*}
Trace preservation and Hilbert-Schmidt isometry imply equality of
the traces on the two sides. Their positive difference must vanish:
\begin{eqnarray*}
\Phi(\bsZ^\dagger \bsZ) = \Phi(\bsZ)^\dagger\Phi(\bsZ) \qquad
\text{for all }\bsZ.
\end{eqnarray*}
The multiplicative-domain criterion then makes $\Phi$ a unital
$*$-homomorphism [7]. A unital $*$-homomorphism from $M_d$ to itself
is a unitary conjugation. This proves Eq.~\eqref{eq:21}.

\subsection{The qubit exception}

For $d=2$, $\Tr{\bsX^3}=0$ for all $\cX\in\Herm_0(\bbC^2)$, so the
cubic restriction is absent. Write a qubit state as
$\rho=\frac12(\I_2+\bsr\cdot\boldsymbol\sigma)$. Eq.~\eqref{eq:28}
becomes $\bsr\mapsto c \bsO\bsr$, where $\bsO\in \rO(3)$. If $\bsO$
is a proper rotation, it is induced by a unitary. If $\bsO$ is
improper, then $-\bsO$ is a proper rotation, and its sign can be
absorbed into a real parameter $\bsa$. Thus the channel has the form
Eq.~\eqref{eq:22}, with complete positivity imposing
$-\frac13\leqslant a\leqslant1$. Conversely,
\begin{eqnarray}\label{eq:34}
S(\rho) = h_2\Pa{\frac{1+\abs{\bsr}}{2}},\quad h_2(x)=-x\ln x-(1-x)\ln(1-x).
\end{eqnarray}
A channel of the form
Eq.~\eqref{eq:22} transforms the Bloch radius as
$\abs{\bsr}\longmapsto \abs{a}\abs{\bsr}$. Equal input entropies
therefore give equal output entropies.

Replacement and unitary channels plainly preserve entropy levels.
Global preservation implies local preservation, completing the proof
of Theorem~\ref{th:2}. \qed

%====================================================%
\section{Entropy ordering and auxiliary-system stability}
%====================================================%

\subsection{Preservation of entropy order}

Consider the stronger-looking ordering condition
\begin{eqnarray}\label{eq:35}
S(\rho)\leqslant S(\sigma) \quad\Longrightarrow\quad
S(\Phi(\rho))\leqslant S(\Phi(\sigma)).
\end{eqnarray}

\begin{cor}[Entropy-order classification]
A channel preserves the ordering in Eq.~\eqref{eq:35} if and only if
it belongs to the classes in Theorem~\ref{th:2}.
\end{cor}

\begin{proof}
Equal input entropies satisfy the premise in both directions, so
Eq.~\eqref{eq:35} implies entropy-level preservation.

Conversely, replacement channels give a constant output entropy, and
unitary channels preserve entropy exactly. For qubits, both the
input entropy and the output entropy under Eq.~\eqref{eq:22} are
decreasing functions of the input Bloch radius. Hence
Eq.~\eqref{eq:35} holds.
\end{proof}
This corollary emphasizes the distinction between entropy ordering
and majorization. For $d\geqslant3$, the entropy-order-preserving
class is strictly smaller than the majorization-preserving class.
For qubits, the two classes coincide.

A further consequence recovers the equal-dimension form of exact
entropy preservation \cite{Zhang2011}.

\begin{cor}[Exact entropy preservation]
If
\begin{eqnarray}\label{eq:36}
S(\Phi(\rho))=S(\rho) \qquad \text{for every }\rho\in\rD(\bbC^d),
\end{eqnarray}
then $\Phi$ is a unitary channel.
\end{cor}

\begin{proof}
Theorem~\ref{th:2} applies. Replacement channels cannot preserve all
input entropy values. In the qubit family, preservation of
pure-state entropy requires $\abs{a}=1$; the complete-positivity
interval leaves only $a=1$.
\end{proof}

\subsection{Extension by an idle auxiliary qubit}

The additional qubit family is not stable under extension by an
unaffected auxiliary system.

\begin{cor}[Ancilla-stable entropy-level preservation]
Let $d\geqslant2$. If
\begin{eqnarray}\label{eq:37}
\Phi\ot\operatorname{id}_2
\end{eqnarray}
preserves entropy equality for all joint states of the system and an
auxiliary qubit, then $\Phi$ is unitary. The same conclusion follows if the condition is imposed only near the maximally mixed joint state.
\end{cor}

\begin{proof}
The joint input and output dimension is $2d\geqslant4$. By
Theorem~\ref{th:2}, the extended channel is either a replacement
channel or a unitary channel.

It cannot be a replacement channel, because the auxiliary marginal
is unchanged and can vary with the input. Hence the extended channel
is unitary and has Choi rank one.

The Choi rank of a tensor product is the product of the Choi ranks,
and the identity channel has Choi rank one. Therefore $\Phi$ itself
has Choi rank one. Its single Kraus operator is unitary because the
channel is trace preserving and the input and output dimensions are
equal.
\end{proof}

The assumption here concerns arbitrary joint states, not merely a
selected family of product inputs. It distinguishes a property of an
isolated device from a property that remains valid when the device
is embedded in a larger quantum system.

%==============================================%
\section{An entropy-matched qutrit witness}
%==============================================%

\subsection{Construction of the input pair}

The classification has a direct operational consequence: nontrivial
qutrit depolarizing noise can separate the entropies of inputs whose
entropies are initially identical.

Define
\begin{eqnarray}\label{eq:38}
\bsA=\frac1{\sqrt2}\diag(1,-1,0), \qquad
\bsB=\frac1{\sqrt6}\diag(2,-1,-1).
\end{eqnarray}
Then $\Tr{\bsA^2}=\Tr{\bsB^2}=1$, whereas
\begin{eqnarray}\label{eq:39}
\Tr{\bsA^3}=0, \qquad \Tr{\bsB^3}=\frac1{\sqrt6}.
\end{eqnarray}
For sufficiently small $t>0$, consider
\begin{eqnarray}\label{eq:40}
\rho_t=\tau_3+t\bsA, \qquad \sigma_t=\tau_3+s(t)\bsB,
\end{eqnarray}
where $s(t)>0$ is determined by the exact condition
\begin{eqnarray}\label{eq:41}
S(\rho_t)=S(\sigma_t).
\end{eqnarray}
Their eigenvalue distributions are
$$
p(t)=\Pa{\frac13+\frac t{\sqrt{2}}, \frac13-\frac t{\sqrt{2}},\frac13},\quad q(t)=\Pa{\frac13+\frac{2s}{\sqrt{6}},\frac13-\frac s{\sqrt{6}},\frac13-\frac s{\sqrt{6}}}.
$$
For $t\in(0,\frac{\sqrt{2}}3)$, the component of $p(t)$ are positive. The entropy $\rH(q(s))$, where $\rH(x)=-\sum_j x_j\ln x_j$, decreases continuously and strictly from $\ln3$ to zero as $s$ increases from zero to $\frac{\sqrt{6}}3$. Entropy decreases strictly along each positive ray away from
$\tau_3$, so a unique small positive solution exists. Expanding
Eq.~\eqref{eq:41} using Eq.~\eqref{eq:29} gives
\begin{eqnarray}\label{eq:42}
s(t) = t+\frac{t^2}{2\sqrt6}+O(t^3).
\end{eqnarray}
This asymptotic expression identifies the mechanism, but finite-parameter predictions and state preparation should use the exact matching condition.

\subsection{Cubic splitting under depolarizing noise}

Now apply the qutrit depolarizing channel
\begin{eqnarray}\label{eq:43}
\cD_a(\rho) = a\rho+(1-a)\tau_3, \qquad -\frac18\leqslant
a\leqslant1.
\end{eqnarray}

\begin{prop}[Cubic entropy splitting]\label{prop:1}
For the exactly entropy-matched inputs in
Eqs.~\eqref{eq:40}-\eqref{eq:41},
\begin{eqnarray}\label{eq:44}
\Delta_a(t) = S(\cD_a(\rho_t)) - S(\cD_a(\sigma_t)) =
\frac{3}{2\sqrt6}a^2(1-a)t^3 + O(t^4).
\end{eqnarray}
\end{prop}

\begin{proof}
For $d=3$, Eqs.~\eqref{eq:29} and ~\eqref{eq:39}
give
\begin{eqnarray*}
S(\tau_3+t\bsA)&=&\ln3-\frac32 t^2+O(t^4),\\
S(\tau_3+s\bsB) &=& \ln3-\frac32s^2+\frac{3}{2\sqrt6}s^3+O(s^4).
\end{eqnarray*}
Matching the input entropies yields Eq.~\eqref{eq:42}. After
depolarization, $t$ and $s$ are replaced by $at$ and $as$. Therefore
$$
\Delta_a(t) = \frac32a^2(s(t)^2-t^2) - \frac{3}{2\sqrt6}a^3s(t)^3
+ O(t^4).
$$
Using $s(t)^2-t^2=\frac{t^3}{\sqrt6}+O(t^4)$ and $s(t)^3=t^3+O(t^4)$
proves Eq.~\eqref{eq:44}.
\end{proof}

For every fixed admissible $a\neq0,1$, the coefficient in Eq.~\eqref{eq:44} is
positive. Thus the output entropies differ for sufficiently small
positive $t$. At $a=0$, the channel is a replacement channel; at
$a=1$, it is the identity. 

Within $a\in[0,1]$, the leading coefficient is maximal at $a=\frac23$:
$$
\Delta_{\frac23}(t) = \frac2{9\sqrt6}t^3+O(t^4).
$$
This is an asymptotic optimum, not a claim about the optimum for every finite $t$.

Because both probes are diagonal, the same calculation applies to the transpose-depolarizing family
\begin{eqnarray}\label{eq:45}
\widetilde{\cD}_a(\rho) = a\rho^\mathsf{T}+(1-a)\tau_3, \qquad
-\frac12\leqslant a\leqslant\frac14,
\end{eqnarray}
because both inputs are diagonal. A fixed output unitary leaves the
entropy difference unchanged.

\subsection{Exact finite-parameter prediction}

Let $u=(\frac13,\frac13,\frac13)$. Without truncating the entropy expansion, the predicted splitting is
\begin{eqnarray}\label{eq:mixuniform}
\Delta_a(t)=\rH(ap(t)+(1-a)u) - \rH(aq(s(t))+(1-a)u),
\end{eqnarray}
where $s(t)$ is obtained from $\rH(p(t))=\rH(q(s(t)))$. The second and third output moments are
\begin{eqnarray*}
m^{(\rho)}_2 = \frac13+a^2t^2,\quad m^{(\sigma)}_2= \frac13+a^2s(t)^2,\\
m^{(\rho)}_3 = \frac19+a^2t^2,\quad m^{(\sigma)}_3= \frac19+a^2s(t)^2 + \frac1{\sqrt{6}}a^3 s(t)^3.
\end{eqnarray*}
The input entropy match compensates unequal cubic information through a small change in the quadratic contribution. Noise rescales these contributions by different powers of a, destroying the compensation.

The two probes are not on the same unitary orbit. They test entropy-level preservation within specified channel families, rather than the single-orbit condition of Theorem~\ref{th:1}.

%==============================================%
\section{Proposed three-mode photonic implementation}
%==============================================%

\subsection{State preparation and noise realization}

Encode a single-photon qutrit in three orthogonal modes, $\ket{0},\ket{1},\ket{2}$. The modes may be spatial paths or time bins. The diagonal inputs Eq.~\eqref{eq:44} and Eq.~\eqref{eq:45} can be prepared by routing a heralded photon into the modes with the prescribed probabilities. The routing record is discarded when describing the encoded qutrit.

Alternatively, a coherent preparation followed by calibrated dephasing may be used, provided that residual coherence is included in the state-characterization uncertainty.

The parameter $s$ should be calibrated through the exact equation $\rH(p(t))=\rH(q(s))$, rather than solely through Eq.~\eqref{eq:42}. Because the ideal output signal is cubic in $t$, even a small residual input-entropy mismatch must be reported alongside an observed output difference.

The use of diagonal probes simplifies the test: neither input coherence nor entanglement is required to observe the predicted splitting. Accordingly, the effect should not be described as a witness of nonclassicality.

\subsection{Randomized Weyl realization}

The depolarizing channel admits an explicit randomized-unitary
realization throughout its complete-positivity interval. Define the qutrit Weyl operators
$$
\bsX_3\ket{j}=\ket{j+1\pmod3}, \quad \bsZ_3\ket{j}=e^{\mathrm{i}
\frac{2\pi j}3}\ket{j},\quad \bsW_{pq}=\bsX_3^p\bsZ_3^q
$$ 
for $p,q\in\set{0,1,2}$. The Weyl twirl satisfies
$$
\frac19\sum^2_{p,q=0}\bsW_{pq}\rho\bsW^\dagger_{pq}=\tau_3.
$$
Choose
\begin{eqnarray}\label{eq:48}
w_{00}=\frac{1+8a}{9}, \qquad w_{pq}=\frac{1-a}{9} \quad\text{for
}(p,q)\neq(0,0).
\end{eqnarray}
Then
\begin{eqnarray}\label{eq:49}
\cD_a(\rho) = \sum_{p,q=0}^2 w_{pq}\bsW_{pq}\rho \bsW_{pq}^\dagger.
\end{eqnarray}
All probabilities in Eq.~\eqref{eq:48} are nonnegative precisely for
$-\frac18\leqslant a\leqslant1$. Randomized mode permutations and phase shifts therefore realize the channel throughout its complete-positivity interval. The selected operation is applied on each trial, and its label is discarded when forming the output state. Multiport interferometers provide a route to more general preparation or analysis transformations \cite{Reck1994,Clements2016}.

For diagonal inputs, phase shifts do not change populations. Summing over the phase index reduces the population transformation to random cyclic permutations with probabilities
$$
v_0=\frac{1+2a}3,\quad v_2=v_3=\frac{1-a}3.
$$
This reduced implementation reproduces the witness populations. Agreement on those populations alone, however, does not verify the action of the full depolarizing channel on coherent inputs.

\subsection{Population and spectral-moment readout}

In the ideal implementation, both outputs remain diagonal. Their entropies can therefore be calculated directly from measured mode populations.

If a fixed unknown output unitary is present, calibrated tomography or basis-independent spectral-moment measurements provide alternatives. For any qutrit state $\omega$, let $m_2=\Tr{\omega^2}$ and
$m_3=\Tr{\omega^3}$. Newton's identities give
\begin{eqnarray*}
\det(\lambda \I_3-\omega) = \lambda^3-\lambda^2
+\frac{1-m_2}{2}\lambda -\frac{1-3m_2+2m_3}{6}.
\end{eqnarray*}
Thus $m_2$ and $m_3$, together with unit trace, determine the eigenvalues and entropy. These moments have multicopy representations \cite{Ekert2002}:
\begin{eqnarray}\label{eq:swaps}
m_2= \Tr{\bsV_2\omega^{\ot2}},\quad m_3=
\Tr{\frac{\bsV_3+\bsV_3^\dagger}{2} \omega^{\ot3}},
\end{eqnarray}
where $\bsV_2$ swaps two copies and $\bsV_3$ cyclically permutes three copies.

Eq.~\eqref{eq:swaps} identifies the observables; it does not, by itself, specify a deterministic passive-linear-optical measurement. An optical implementation must separately address copy preparation, photon indistinguishability, the realized interference observable and detector response.

For direct population measurements on a diagonal state, $N$ independent trials give unbiased estimators
\begin{eqnarray}
\widehat m_k = \frac{\sum^2_{j=0}(n_j)_k}{(N)_k},\quad k=2,3,
\end{eqnarray}
where $n_j$ is the count in mode $j$ and $(x)_k=x(x-1)\cdots (x-k+1)$.

These expressions assume normalized outcomes within the encoded qutrit space. Finite-sample moment estimates need not correspond exactly to a physical spectrum, so entropy reconstruction should propagate statistical uncertainty and enforce positivity and unit trace where needed.

\subsection{Preparation mismatch and readout uncertainty}

A model-based error budget can be obtained for the calibrated diagonal preparations.
Let $s_*=s(t)$ be the exactly matched value, while the prepared second input uses $s$. Define 
\begin{eqnarray*}
\delta_{\op{in}} = \Abs{\rH(p(t)) - \rH(q(s))},\quad \rH_{q,a}(x)=\rH(aq(x)+(1-a)u).
\end{eqnarray*}
Direct differentiation gives
$$
\rH'_{q,a}(x) = -\frac{2a}{\sqrt{6}}\ln \Pa{\frac{\frac13+\frac{2ax}{\sqrt{6}}}{\frac13-\frac{ax}{\sqrt{6}}}}.
$$
Suppose $s$ and $s_*$ lie in a positive interval $I$ on which all relevant populations are strictly positive. Set
\begin{eqnarray}\label{eq:ratios}
K_a(I) = \sup_{x\in I}\frac{\Abs{\rH'_{q,a}(x)}}{\Abs{\rH'_{q,1}(x)}}.
\end{eqnarray}
Since $\rH_{q,1}$ is strictly decreasing on $I$, comparison of the two integrals between $s$ and $s_*$ gives
$$
\abs{\rH_{q,a}(x) - \rH_{q,a}(s_*)}\leqslant K_a(I) \delta_{\op{in}}.
$$
For fixed $a$, the ratio inside Eq.~\eqref{eq:ratios} tends to $a^2$ as $x\to0^+$. Thus, within this preparation model, the output bias due to a small entropy-matching error is controlled approximately by $a^2\delta_{\op{in}}$ near the maximally mixed state.

Suppose next that each reconstructed output state differs from the actual output by trace distance at most $\epsilon$, where $\epsilon\in(0,\frac23]$. The entropy continuity bound \cite{Audenaert2007} gives
$$
\Abs{S(\widehat \omega) - S(\omega)}\leqslant F_3(\epsilon)=h_2(\epsilon)+\epsilon\ln2.
$$
If $\widehat \Delta$ is the entropy difference reconstructed from the two outputs, then
\begin{eqnarray}\label{eq:endiff}
\Abs{\widehat\Delta - \Delta_a(t)}\leqslant K_a(I)\delta_{\op{in}}+2F_3(\epsilon),
\end{eqnarray}
where $\Delta_a(t)$ is the ideal, exactly matched prediction Eq.~\eqref{eq:mixuniform}.

The bound is conditional on the specified preparation and channel model. It is not a model-independent certification criterion for an arbitrary unknown process. Experimental confidence intervals must also be obtained from the measurement statistics rather than inferred from Eq.~\eqref{eq:endiff} alone.

\subsection{Input-independent optical loss}

Loss may be represented without postselection by retaining a vacuum outcome. Consider
\begin{eqnarray}\label{eq:vac}
\cL_{\eta,a}(\rho) = \eta\cD_a(\rho) \oplus (1-\eta)\proj{\op{vac}},\quad \eta\in[0,1],
\end{eqnarray}
with an input-independent survival probability  $\eta$.

The direct-sum structure gives $S(\cL_{\eta,a}(\rho))=h_2(\eta)+\eta S(\cD_a(\rho))$ and therefore
\begin{eqnarray}\label{eq:scalprod}
\Delta_{\eta,a}(t) = \eta\Delta_a(t).
\end{eqnarray}
Under this model, loss reduces the unconditional splitting linearly but does not change its sign for $\eta>0$.

The map Eq.~\eqref{eq:vac} has a three-dimensional input and a four-dimensional output. It is included as an optical loss model, not as an instance of the equal-dimension classification. Conditioning on survival recovers $\cD_a$ only under the stated input-independent loss assumption. Mode-dependent or input-dependent loss generally changes the normalized transformation and must be modelled separately.

\subsection{Scope of the proposed test}

The probe pair is an explicit witness for the depolarizing and transpose-depolarizing channel families; a single pair does not characterize an arbitrary unknown channel.

Likewise, Theorem 1 requires moment constancy on an entire unitary orbit. Restricting attention to two moment orders does not make the orbit condition a finite-setting experimental certification protocol.

Finally, encoded qutrit dynamics must be distinguished from photon loss, detector response and input-dependent conditioning. No experimental data or experimentally validated detection threshold is reported here. Equations Eq.~\eqref{eq:mixuniform}, Eq.~\eqref{eq:endiff} and Eq.~\eqref{eq:scalprod} provide analytical predictions and bounds for assessing a future implementation.

%================================================%
\section{Discussion and conclusions}
%================================================%

The results establish a hierarchy between two channel preservation properties. Majorization preservation is equivalent to preservation of spectral equivalence and can be identified from the second and third output moments on the unitary orbit of one simple-spectrum input. The quadratic condition forces every nonreplacement channel to be unital and isotropic on traceless perturbations. The cubic condition then identifies a unitary or transpose-unitary Jordan structure. For qubits, the quadratic condition alone suffices.

Entropy-level preservation is more restrictive. In equal dimensions $d\geqslant3$, even a local condition near the maximally mixed state permits only replacement and unitary channels. The entropy Hessian imposes isotropic action, while the cubic entropy term excludes nontrivial contraction. Qubits are exceptional because their cubic trace invariant vanishes. Requiring preservation after extension by an idle auxiliary qubit eliminates the nonunitary exception.

The qutrit construction makes the difference operationally explicit. Two exactly entropy-matched diagonal inputs acquire an output entropy difference with leading behaviour $a^2(1-a)t^3$  under nontrivial depolarizing noise. Three-mode preparation, randomized Weyl operations and population or spectral-moment readout provide a proposed photonic setting for testing this prediction.

Quantitative stability of the classification remains open: approximate moment or entropy-level preservation should be related to distance from the rigid channel families. Degenerate input orbits and unequal input and output dimensions also require separate analyses. In the latter case, injectivity no longer forces surjectivity, and isometric embeddings already enlarge the exactly entropy-preserving class.

Quadratic and cubic spectral information thus provide a common mechanism behind majorization-preserver normal forms and entropy-level rigidity, while the three-mode witness reveals the distinction between them.

\subsection*{Declaration on the use of generative AI}

The authors designed the study, carried out the mathematical analysis, and prepared the initial
manuscript. Generative-AI tools were used for language refinement and assistance with computational
checks. All mathematical statements and computations were independently reviewed
and verified by the authors, who takes full responsibility for the content of the manuscript.

\subsection*{Data availability}

No experimental or numerical datasets were generated in this
theoretical study.

\subsection*{Funding}

L.L. was supported by the Research Startup Funding of Hangzhou City
University (Grant No. QD2024012) and by the National Natural Science
Foundation of China (Grant No. 12271474).

%------------------------------------------------------------%
\appendix
\appendixpage
\addappheadtotoc
%------------------------------------------------------------%

%========================================================================%
\section{Low-degree polynomials on a simple-spectrum orbit}\label{app:A}
%========================================================================%

We prove Lemma~\ref{lem:1} using the algebraic structure of a
simple-spectrum conjugacy orbit. Complexify $\Herm_0(\bbC^d)$ to
$\mathfrak{sl}_d(\bbC) = \set{\bsX:\Tr{\bsX}=0}$, and write
$p_k(X)=\Tr{\bsX^k}$ and $\kappa_k=p_k(\bsA)$. Because $\bsA$ has
distinct eigenvalues, the equations
\begin{eqnarray}\label{eq:A1}
p_k(\bsX)=\kappa_k, \qquad k=2,\ldots,d,
\end{eqnarray}
determine its complex conjugacy orbit. Indeed, Newton's identities
determine the characteristic polynomial from these power sums and
the zero trace. The characteristic polynomial has distinct roots, so
every matrix satisfying Eq.~\eqref{eq:A1} is diagonalizable and
conjugate to $\bsA$.

The unitary orbit $\cO_{\bsA}$ is Zariski dense in this complex
orbit. Moreover, the differentials of $p_2,\ldots,p_d$ are
independent at every point of the fibre Eq.~\eqref{eq:A1}. To see
this, a dependence among the differentials would imply that a
polynomial of degree at most $d-1$ in $\bsX$ is scalar. Evaluation
on the $d$ distinct eigenvalues forces all coefficients of the
dependence to vanish.

The fibre is therefore a smooth complete intersection, and its
vanishing ideal is
\begin{eqnarray}\label{eq:A2}
I_{\bsA} = (p_2-\kappa_2,\ldots,p_d-\kappa_d).
\end{eqnarray}
We also need the degree property
\begin{eqnarray}\label{eq:A3}
F\in I_{\bsA},\quad \deg F\leqslant r \quad\Longrightarrow\quad
F=\sum_{k=2}^d h_k(p_k-\kappa_k), \qquad \deg h_k+k\leqslant r.
\end{eqnarray}
For completeness, the homogeneous polynomials \(p_2,\ldots,p_d\)
form a regular sequence. Their common zero set is the nilpotent
cone, which has codimension \(d-1\) in \(\mathfrak{sl}_d(\mathbb
C)\), as follows from the dimensions of the nilpotent Jordan orbits.
Since the polynomial ring is Cohen-Macaulay, the required
regular-sequence property follows.

In any representation of \(F\) by the generators in
Eq.~\eqref{eq:A2}, cancellation above degree \(r\) gives a
homogeneous syzygy among \(p_2,\ldots,p_d\). Such syzygies are
generated by the Koszul relations. Subtracting the corresponding
relations among \(p_k-\kappa_k\) lowers the maximal degree.
Repeating this procedure yields Eq.~\eqref{eq:A3}. Now suppose
$g+\ell=C$ on $\cO_{\bsA}$, where $g$ is homogeneous quadratic and
$\ell$ is linear. By Eq.~\eqref{eq:A3},
$g+\ell-C=\gamma(p_2-\kappa_2)$. Comparison of homogeneous terms
gives $g=\gamma p_2$ and $\ell=0$. This proves the first statement.

For the second statement, let $f$ be homogeneous cubic and equal to
$C$ on the orbit. Equation Eq.~\eqref{eq:A3} gives
\begin{eqnarray}\label{eq:A4}
f-C = (\ell+b)(p_2-\kappa_2) + \beta(p_3-\kappa_3),
\end{eqnarray}
where $\ell$ is linear and $b$ is constant.

Comparison of degree-two terms gives $b=0$. Comparison of degree-one
terms gives $\kappa_2\ell=0$. Since $\bsA\neq0$ is Hermitian,
$\kappa_2=\Tr{\bsA^2}>0$, and therefore $\ell=0$. The degree-three
part of Eq.~\eqref{eq:A4} is $f=\beta p_3$. This proves
Lemma~\ref{lem:1}. \qed

%=============================================================================%

%=============================================================================%

\end{document}